\newif\ifcamera
\cameratrue
\camerafalse %

\ifcamera
\documentclass[twoside,leqno]{article}
\else
\documentclass[11pt]{article}
\fi
\newif\ifsubmission
\submissiontrue
\submissionfalse %

\newif\ifcomment
\commenttrue

\ifsubmission
\commentfalse
\author{Anonymous Authors}
\else
\ifcamera
\author{
  Yotam Kenneth-Mordoch%
  \thanks{Weizmann Institute of Science, Rehovot, Israel (\email{yotam.kenneth@weizmann.ac.il}).}
}
\else
\author{
  Yotam Kenneth-Mordoch
  \\ Weizmann Institute of Science
  \\ \texttt{yotam.kenneth@weizmann.ac.il}
}
\fi

\fi
\ifcamera
\usepackage[letterpaper]{geometry}
\usepackage{siamproceedings}
\else
\usepackage[letterpaper, margin=1in]{geometry}
\usepackage{amsthm}
\usepackage{amsmath}
\fi
\usepackage{algorithm}
\usepackage[noend]{algpseudocode}

\usepackage{amsfonts}
\usepackage{import} 
\usepackage{xifthen}
\usepackage{mathtools}
\usepackage{arydshln}%
\usepackage{dsfont}
\usepackage[cal=esstix]{mathalfa}
\usepackage{thm-restate}

\usepackage{comment}
\usepackage[colorinlistoftodos,textsize=tiny,textwidth=2cm,color=red!25!white,obeyFinal]{todonotes}

\ifcamera\else %
\usepackage[hypertexnames=false]{hyperref}
\hypersetup{colorlinks=true,allcolors=blue}
\usepackage{cleveref}
\fi

\AddToHook{cmd/appendix/before}{\crefalias{section}{appendix}}

\newcommand{\mintcut}{\mathrm{cut}}

\newcommand{\tO}{\tilde{O}}

\DeclareMathOperator{\polylog}{polylog}
\newcommand{\indic}[1]{ \mathds{1}_{\{#1\}} } %
\providecommand{\set}[1]{{\{#1\}}}
\newcommand{\eqdef}{\coloneqq}

\ifcomment

\else
\fi

\ifcamera
\newsiamthm{question}{Question}
\newsiamthm{claim}{Claim}
\newsiamthm{conjecture}{Conjecture}
\newsiamthm{observation}{Observation}
\newsiamthm{fact}{Fact}
\newsiamremark{remark}{Remark}
\providecommand{\qedhere}{}
\else
\newtheorem{theorem}{Theorem}[section]

\newtheorem{claim}[theorem]{Claim}

\newtheorem{lemma}[theorem]{Lemma}

\newtheorem{observation}[theorem]{Observation}

\fi

\title{Deterministic Global Minimum Cut Using $\tilde{O}(n^{3/2})$ Cut Queries}
\date{}
\begin{document}

\maketitle

\ifcamera
\fancyfoot[R]{\scriptsize{Copyright \textcopyright\ 2027 by SIAM\\
Unauthorized reproduction of this article is prohibited}}
\fi
\begin{abstract}  
  Finding a global minimum cut in an unweighted graph is a fundamental task in the cut-query model.
  While there exists a randomized algorithm for this task using $\tilde O(n)$ queries, known deterministic algorithms use $\tilde O(n^{5/3})$ queries.
  We improve upon this result and give a simple deterministic algorithm that finds a global minimum cut using $\tilde O(\min\left\{ n\delta,n^2/\delta \right\})$ cut queries, which equals $\tO(n^{3/2})$ in the worst case, where $\delta$ is the minimum degree of the graph.
  The key ingredient in our approach is an efficient maximum $s,t$-flow algorithm, for the case when the minimum $s,t$-cut has value at most $\delta$.
\end{abstract}

\section{Introduction}
\label{sec:introduction}
Finding a global minimum cut is a fundamental problem in graph algorithms. 
Given an undirected unweighted graph $G=(V,E)$ on $n$ vertices, the goal is to find a non-empty set $S\subset V$ that minimizes
\[
  \mintcut_G(S) 
  \eqdef |E(S,V\setminus S)|,
\]
where $E(S,T)$ is the set of edges with one endpoint in $S$ and the other in $T$.
The value of the minimum cut is denoted by $\lambda_G$.
In this work, we focus on finding a global minimum cut in the \emph{cut-query model} \cite{RSW18}.
In this model, algorithms can access an input graph $G=(V,E)$ only via cut queries that report $\mintcut_G(S)$ for a given set $S\subseteq V$, that is, the number of edges crossing the cut between $S$ and $V\setminus S$.
The goal in this model is to minimize the \emph{query complexity}, i.e. the number of queries that the algorithm makes to the oracle.

The randomized query complexity of finding a global minimum cut is fairly well understood.
In the first work to consider the model, it was shown that a minimum cut of an unweighted graph can be found using $\tO(n)$ randomized queries \cite{RSW18}, where $\tO(\cdot)$ hides polylogarithmic factors in $n$.
This was later improved along several directions.
First, an algorithm for finding a minimum cut in a weighted graph using the same query complexity was provided in \cite{MN20}.
Subsequently, it was shown that $O(n)$ randomized cut queries suffice in the case of unweighted graphs \cite{AEGLMN22}, shaving the polylogarithmic factors.
Finally, a randomized algorithm for finding a global minimum cut using only $2$ rounds of queries, albeit with increasing the query complexity to $\tO(n^{4/3})$ was given in \cite{KK25}.
Moreover, the randomized query complexity can be matched by pseudodeterministic algorithms, which output a fixed canonical minimum cut with high probability~\cite{AV26,Ken26}.

In comparison, the deterministic complexity of finding a global minimum cut is less well understood.
The only deterministic algorithm known, which uses $\tilde O(n^{5/3})$ cut queries was given in \cite{ASW25}.
A main tool in this algorithm is a version of Dinitz's max-flow algorithm~\cite{Dinitz06}, for flows of value at most $O(n)$, using $\tO(n^{5/3})$ cut queries. 
It then leverages this primitive to apply the framework of \cite{LP20} for finding a minimum cut using $\polylog n$ max flows.
This framework relies on complex machinery, including terminal sparsification, and expander decomposition.
Furthermore, applying it in the cut-query setting requires great care as the max-flow primitive of \cite{ASW25} is limited to flows of value $O(n)$.

Our main result is a simple deterministic algorithm that finds a global minimum cut using $\tO(\min\left\{ n\delta,n^2/\delta \right\})$ cut queries, where $\delta$ is the minimum degree of $G$.
The worst-case complexity of our algorithm is $\tO(n^{3/2})$, which improves upon the query complexity of \cite{ASW25}.
Furthermore, it achieves even better query complexity when $\delta$ is small or large.
While it was known before that one can achieve better query complexity when $\delta$ is small, by packing $\delta$ edge disjoint trees, we give the first indication that finding a global minimum cut becomes easier when the minimum degree is large.
In particular, when $\delta = \Omega(n)$ our algorithm matches the randomized complexity of finding a global minimum cut (up to polylogarithmic factors).
We are now ready to state our main result.

\begin{theorem}
\label{thm:intro-main}
There exists a deterministic algorithm that, given cut-query access to a simple graph $G$ on $n$ vertices with minimum degree $\delta$, finds a global minimum cut of $G$ using $\tilde O(\min\left\{ n\delta, n^2/\delta \right\})$ cut queries.
\end{theorem}

\subsection{Future Work}
\paragraph{Faster Global Minimum Cut.}
The main open question from this work is what is the deterministic query complexity of finding a global minimum cut.
In particular, algorithms using $O(n)$ cut queries and matching existing randomized algorithms \cite{RSW18,MN20,AEGLMN22} or a lower bound of $\omega(n)$ would be of great interest.

\paragraph{Deterministic Maximum Flow.}
Another natural question in the deterministic cut-query model is whether one can solve minimum $s,t$-cut using $o(n^{5/3})$ cut queries.
Maximum flow is a central primitive both in \cite{ASW25} and in this work (although in a restricted setting), and it is natural to ask whether one can improve the query complexity of this primitive.
Recent randomized algorithms for this problem use $\tO(n^{3/2})$ cut queries \cite{JNS26,KK26}, and it is open whether this can be achieved deterministically.

\section{Global Minimum Cut Algorithm}
\label{sec:tech-overview}
Our main technical contribution is an efficient max-flow algorithm that uses $\tO(n)$ cut queries when the minimum $s,t$ cut has value at most $\delta$.
This is based on the fact that in this case, there exists a maximum flow between $s$ and $t$ that is supported on at most $O(n)$ edges.
This was previously observed in \cite{GY95,HKR97}, however it was not known how to exploit this property to obtain an efficient algorithm in the cut-query model.
It is also spiritually related to previous work bounding the number of edges in a maximum $s,t$-flow in a graph \cite{KL98,RSW18}.
We manage to leverage this property to give an efficient max-flow algorithm for this case.

To gain some intuition, let $s,t$ be vertices separated by a global minimum cut and let $f$ be a maximum $s,t$-flow obtained by applying the Edmonds-Karp algorithm \cite{EK72}, which augments the flow along shortest paths in the residual graph.
Recall that the radius of a connected graph with minimum degree $\delta$ is at most $O(n/\delta)$ \cite{EPPT89}.
Morally, the same short radius property holds in the residual graph as long as the flow does not remove too many edges;
hence the length of each augmenting path is $O(n/\delta)$, and the total number of edges carrying flow is at most $O(n)$.
Making this precise requires more care and is given later in \Cref{sec:bounded-dinitz}.
Unfortunately, this intuition does not immediately yield an efficient algorithm, since we do not have efficient decremental BFS algorithms.
Instead, our algorithm augments along short ``enough'' augmenting paths.
Specifically, we show that if one is given any augmenting path, then it is possible to find a short augmenting path of length at most $O(n/\delta)$ using $\tO(n/\delta)$ cut queries.
This allows us to implement the following bounded max-flow algorithm, whose proof is provided in \Cref{sec:bounded-dinitz}.
\begin{theorem}[Bounded Max-Flow]
\label{thm:bounded-dinitz}
There exists a deterministic algorithm that, given cut-query access to a simple graph $G$ on $n$ vertices and
vertices $s,t\in V$, computes an exact $s,t$ max-flow of value at most $\delta$, or certifies that the $s,t$ max-flow value is larger than $\delta$, using $\tO(n)$ queries.
If the algorithm finds a maximum flow, it also returns a minimum $s,t$-cut.
\end{theorem}

Notice that using the above primitive, it is possible to find a global minimum cut in $\tO(n)$ cut queries if given two vertices $s,t$ separated by a global minimum cut.
To find such a pair of vertices, we use the following results of \cite{ASW25}.
\begin{lemma}[Lemma 5.3 of \cite{ASW25}]
\label{lem:dominating-global}
Let $G$ be a simple graph on $n$ vertices, and let $R$ be a dominating set of $G$.
If $\lambda_G<\delta$, then every global minimum cut has a vertex of $R$ on each side.
\end{lemma}
\begin{theorem}[Theorem 5.4 of \cite{ASW25}]
\label{thm:dominating-set}
There exists a deterministic algorithm that, given cut-query access to a simple graph $G$ on $n$ vertices, computes a dominating set $R\subseteq V$ of size $\tO(n/\delta)$ using $\tO(n)$ cut queries.
\end{theorem}
This suggests the following natural algorithm, which entirely bypasses the \cite{LP20} framework leveraged by \cite{ASW25}.
Fix a dominating set $R$ of size $\tilde O(n/\delta)$, and an arbitrary source vertex $s\in R$;
then find the minimum cut among all pairs $(s,r)$ for $r\in R\setminus\{s\}$ using the bounded max-flow primitive.
Unfortunately, the size of the dominating set $\tO(n/\delta)$ can be large when the minimum degree is small.
To handle the latter case, we use a Nagamochi-Ibaraki sparsifier \cite{NI92}.
This result is standard, but we include a proof for completeness in \Cref{sec:ni-sparsifier-proof}.
\begin{lemma}[Nagamochi-Ibaraki Sparsifier]
\label{lem:ni-certificate}
Given a simple graph $G$ on $n$ vertices and an integer $\tau\ge 1$, using $\tO(n\tau)$ cut queries one can explicitly recover a subgraph $H\subseteq G$ with $O(n\tau)$ edges, such that for every cut $S\subseteq V$,
\[
  \mintcut_H(S) 
  \ge \min\left\{ \mintcut_G(S), \tau \right\}
  .
\]
\end{lemma}
\begin{proof}[Proof of Theorem~\ref{thm:intro-main}]
  Begin by querying all singleton cuts, to find the minimum degree $\delta$ of the graph.
  If $\delta=0$, return any isolated vertex.
  If $\delta\le n^{1/2}$, explicitly recover a $2\delta$-NI sparsifier $H$ using \Cref{lem:ni-certificate}, and find a global minimum cut on the sparsifier using some offline algorithm.
  Otherwise, construct a dominating set $R$ of size $\tilde O(n/\delta)$ using \Cref{thm:dominating-set}, and fix an arbitrary source vertex $s\in R$.
  For every $r\in R\setminus\{s\}$, run \Cref{thm:bounded-dinitz} on the pair $(s,r)$.
  Finally, return the minimum cut found among all pairs $(s,r)$, or the minimum degree if no cut of value at most $\delta$ is found.

  We now argue the correctness and query complexity of the algorithm.
  In the low degree case, the correctness is immediate since the minimum cut is at most $\delta$ and by \Cref{lem:ni-certificate} for every $S\subseteq V$, $\mintcut_H(S) \ge \min\left\{ \mintcut_G(S), 2\delta \right\}$.
  It is also clear that the query complexity is $\tO(n\delta)$, since given the sparsifier no further queries are needed to find a global minimum cut.
  In the high degree case, let $S\subseteq V$ be a global minimum cut. 
  If $\lambda_G \le \delta-1$, then by \Cref{lem:dominating-global} there exists a vertex $r^*\in R$ such that $r^*$ is on the opposite side of the cut $S$ from $s$.
  Hence, $s$ is separated from $r^*$ by a global minimum cut, and the bounded max-flow primitive will find it when called on $(s,r^*)$.
  If $\lambda_G= \delta$, then we return the minimum degree as the global minimum cut, which is correct.
  Finally, the query complexity is $\tO(n^2/\delta)$ as we call \Cref{thm:bounded-dinitz} $\tO(n/\delta)$ times.
\end{proof}

\section{Preliminaries}
\label{sec:preliminaries}
Throughout, let $G=(V,E)$ be a simple undirected graph on $n$ vertices with minimum degree $\delta$. 
Given two (not necessarily disjoint) sets $A,B\subseteq V$, let $E(A,B)$ be the set of edges with one endpoint in $A$ and one endpoint in $B$.
We also use the following primitives.
\begin{claim}
\label{claim:count-crossing-edges}
  Given a simple graph $G$ on $n$ vertices and two disjoint sets $A,B\subseteq V$, one can compute $|E(A,B)|$ using $3$ cut queries.
\end{claim}
\begin{proof}
  The proof follows from the following identity, which holds for any disjoint sets $A,B\subseteq V$, 
  \begin{equation*}
    \mintcut_G(A\cup B) = \mintcut_G(A) + \mintcut_G(B) - 2|E(A,B)|.
    \qedhere
  \end{equation*}
\end{proof}
\begin{claim}
  \label{claim:find-outgoing-edge}
  Given a simple graph $G$ on $n$ vertices, and disjoint vertex sets $S,T\subseteq V$, one can find an edge in $E(S,T)$ or determine that no such edge exists using $O(\log n)$ cut queries.
\end{claim}
\begin{proof}
  The proof follows from a simple binary search procedure. 
  Let $T_0=T$, and for $i=0,1,\ldots,\log |T|$, split $T_i$ into two equal parts $T_i^1,T_i^2$, and find $|E(S,T_i^1)|$ using Claim~\ref{claim:count-crossing-edges}.
  If $|E(S,T_i^1)|>0$, set $T_{i+1}=T_i^1$ as there exists an edge in $E(S,T_i^1)$.
  Otherwise, query $|E(S,T_i^2)|$, and again if $|E(S,T_i^2)|>0$, set $T_{i+1}=T_i^2$.
  Finally, if both cuts have value $0$, terminate the procedure as there is no edge in $E(S,T_i)$.
  After $\log |T|$ iterations, $T_{\log |T|}=\{t\}$ is a singleton and the set $E(S,\set{t})$ is non-empty.
  We then repeat the above procedure with $S$ and $t$ to find a vertex $s\in S$ such that $(s,t)\in E(S,\set{t})$.
  It is easy to see that this procedure uses $O(\log |T| + \log |S|) = O(\log n)$ cut queries, and finds an edge in $E(S,T)$ if one exists.
\end{proof}

\section{A Bounded Max-Flow Algorithm}
\label{sec:bounded-dinitz}
In this section we prove Theorem~\ref{thm:bounded-dinitz}.
Throughout, we use the notion of a residual flow graph.
In a simple graph $G$, a flow $f$ between two vertices $s,t\in V$ is a set of edge-disjoint paths from $s$ to $t$, where each path carries one unit of flow.
The residual graph $G_f$ is obtained by first turning the undirected graph $G$ into a directed graph by replacing every undirected edge $\{u,v\}$ with two directed arcs $(u,v)$ and $(v,u)$.
When augmenting $f$ with a path $\pi\in G_f$, which consists of directed arcs, we perform the following process. 
Let $\pi=(s=p_1,\ldots,p_k=t)$, then for every arc $\{p_j,p_{j+1}\}$, remove the directed arc $(p_j,p_{j+1})$ from $G_f$, and add a copy of the arc $(p_{j+1},p_j)$ to $G_f$.
Note that this can create parallel arcs in $G_f$.
The flow $f$ is composed of all the edges corresponding to the arcs carrying flow.
Finally, let $q(f)$ be the support size of $f$, and denote by $E_f$ the set of reverse arcs in $G_f$ corresponding to edges carrying flow in $f$.
Throughout, we will use the following easy observation.
\begin{observation}
  \label{observation:query-residual-graph}
  Let $G=(V,E)$ be a graph and $F\subseteq E$ be a set of explicitly known edges.
  Then, one can query any cut in $G\setminus F$ using a single cut query in $G$.
\end{observation}
\begin{proof}
  Fix some query $S\subseteq V$.
  Notice that $\mintcut_{G\setminus F}(S) = \mintcut_G(S) - \sum_{e\in F} \indic{e\in S\times (V\setminus S)}$, where $\indic{e\in S\times (V\setminus S)}$ is the indicator function for whether $e$ crosses the cut $S$.
  Since $F$ is explicitly known, this can be computed by a single cut query.
\end{proof}

We use the following result from \cite{KL98}, which shows that the union of a maximal spanning forest of $G\setminus f$ and $E_f$ contains an augmenting path $\pi$ from $s$ to $t$ if and only if one exists in $G_f$.
The proof of this theorem is straightforward and is provided in \Cref{sec:kl98-proof} for completeness.
\begin{theorem}[Theorem 3.4 of \cite{KL98}]
  \label{thm:kl98-exist-augmenting}
  Let $T$ be a maximal spanning forest of $G\setminus f$, where $f$ is some $s,t$ flow in $G$.
  Then, there exists an $s,t$ augmenting path in $T\cup E_f$ if and only if there is an augmenting path in $G_f$.
\end{theorem}
While \Cref{thm:kl98-exist-augmenting} guarantees the existence of an augmenting path, this path does not necessarily have length $O(n/\delta)$.
To find a short augmenting path, we search for shortcutting vertices, which allow us to shorten $\pi$ to the desired length.
Specifically, we show through a simple pigeonhole argument that there exists some vertex $v\in V$ that is a neighbor in $G\setminus f$ of two far apart vertices in $\pi$, and hence we can shortcut $\pi$ through $v$.
The following lemma, which is the main technical component in the proof of \Cref{thm:bounded-dinitz}, shows how to implement this process.
\begin{lemma}
  \label{lem:shortcut-augmenting-path}
  There exists an algorithm that, given a simple augmenting path $\pi$ in $G_f$ where $f$ is a flow supported on $q(f)$ edges, the flow $f$, and the list of vertex degrees, finds a simple augmenting path $\pi'$ of length at most $8\cdot(n+q(f))/\delta$ in $G_f$ using $\tO(n/\delta)$ cut queries.
\end{lemma}
\begin{proof}
  Let $\pi=\set{p_1,\ldots,p_r}$, where the vertices $p_1,\ldots,p_r$ are distinct since $\pi$ is simple, and notice that the length of $\pi$ is $r-1$.
  If $r-1 \le 8\cdot(n+q(f))/\delta$, then we are done.
  We say that $u\in V$ is a \emph{shortcutting vertex} for $\pi$, if there exist two neighbors $p_i,p_j\in \pi$ of $u$ in $G\setminus f$ such that $|i-j|\ge r\delta/(8n)+2$.
  Since  $r-1> 8\cdot(n+q(f))/\delta$, shortcutting $\pi$ through $(p_i,u),(u,p_j)$, the length of $\pi$ becomes at most $(r-1)-(r\delta/(8n)+2)+2\le r(1-\delta/(8n))$, where the $2$ is from adding the edges $(p_i,u),(u,p_j)$ to $\pi$.
  These edges are in $G\setminus f$, and hence their corresponding arcs are in $G_f$, so the shortcut path is an augmenting path in $G_f$.
  Note that if the shortcutting creates a cycle (i.e., $u\in\pi$), we can remove it to obtain a simple path, which only decreases the length.
  Repeating this process $O((n/\delta)\log n)$ times yields an augmenting path of length at most $8\cdot(n+q(f))/\delta$.
  For the rest of the proof we focus on showing that a shortcutting vertex always exists, and that it can be found using $\polylog(n)$ cut queries.

  For every $v\in V$, let $N_{\pi}(v)$ be the set of neighbors of $v$ in $\pi$ in $G\setminus f$, and let $d_\pi(v)=|N_{\pi}(v)|$.
  Notice that if $d_\pi(v) \ge r\delta/(8n)+3$, there exist two neighbors $p_i,p_j\in N_{\pi}(v)$ such that $|i-j|\ge r\delta/(8n)+2$, and $v$ is a shortcutting vertex.
  Notice that $\sum_{v\in V} d_\pi(v) \ge \sum_{p\in \pi} \deg_G(p)-2q(f)\ge r\delta-2q(f)$, where $\deg_G(p)$ is the degree of $p$ in $G$ and the $2$ factor is since each edge in $f$ has two endpoints.
  Therefore, on average, $\tfrac{1}{n}\sum_{v\in V} d_\pi(v) \ge (r\delta-2q(f))/n\ge r\delta/(2n)$ where the last inequality follows from $r\delta \ge 8\cdot(n+q(f)) >8q(f)$.
  Hence, by the pigeonhole principle, there exists some shortcutting $u\in V$ such that $d_\pi(u) \ge r\delta/(2n)>r\delta/(8n)+3$, where the last inequality is since $r\delta > 8\cdot(n+q(f))\ge 8n$.

  To find $u$, apply a binary search approach.
  Partition $V$ into two sets $V_1,V_2$ of roughly equal size, and find $\tfrac{1}{|V_1|}\sum_{v\in V_1} d_\pi(v),\tfrac{1}{|V_2|}\sum_{v\in V_2} d_\pi(v)$.
  Then, recurse on the set with the larger average.
  By the additivity of $d_\pi(v)$, if $\tfrac{1}{|V_i|}\sum_{v\in V_i} d_\pi(v)\ge \tfrac{1}{|V|}\sum_{v\in V} d_\pi(v)$, then there exists $v\in V_i$ with $d_\pi(v) \ge r\delta/(8n)+3$.
  Perform this $O(\log n)$ times until the process culminates in a singleton set, which is a shortcutting vertex $u$.
  
  We now show how to find  $\sum_{v\in V_i} d_\pi(v)$ efficiently using cut queries.
  Denote $V_i = W_i \cup A_i$, where $W_i = V_i\setminus \pi$ and $A_i = V_i\cap \pi$.
  Then, we can write $\sum_{v\in V_i} d_\pi(v) = |E(W_i,\pi)| + |E(A_i,\pi\setminus A_i)| + 2|E(A_i,A_i)|$.
  The first two elements can be computed using $O(1)$ cut queries in $G\setminus f$ by \Cref{claim:count-crossing-edges}.
  In addition, the third element is given by $2|E(A_i,A_i)|=\sum_{p\in A_i} \deg_{G\setminus f}(p) - \mintcut_{G\setminus f}(A_i)$.
  The list of vertex degrees in $G$ is known by the theorem statement, and hence we can compute vertex degrees in $G\setminus f$ by subtracting the number of edges in $f$ that are incident to each vertex.
  Finally, cut queries to $G\setminus f$ can be simulated using a single cut query in $G$ by \Cref{observation:query-residual-graph}.

  Therefore, one can find a shortcutting vertex $u$ using $O(\log n)$ cut queries.
  Then, perform two binary searches to find the first and last neighbor of $u$ in $\pi\setminus\set{u}$ in $G\setminus f$, each using $O(\log r)$ cut queries by \Cref{claim:count-crossing-edges}.
  Hence, the total number of cut queries to find a shortcut is $O(\log n + \log r) = O(\log n)$.
  Since the algorithm finds a shortcut $O(n/\delta \log n)$ times, the total number of cut queries used is $\tO(n/\delta)$.
\end{proof}
\subsection{Proof of Theorem~\ref{thm:bounded-dinitz}}
In this section, we prove Theorem~\ref{thm:bounded-dinitz}.
To do so, we need the following helper lemma for maintaining a maximal spanning forest of $G\setminus f$ while augmenting the flow.
\begin{lemma}
  \label{lem:maintain-maximal-spanning-forest}
  There exists an algorithm that, given a maximal spanning forest $T$ of $G\setminus f$, the flow $f$, and an augmenting path $\pi$ of $G_f$, finds a maximal spanning forest of $G\setminus f'$ using $\tO(|\pi|)$ cut queries,  where $f'$ is the updated flow after augmenting along $\pi$.
\end{lemma}
\begin{proof}
  Let $T' = T\setminus \pi$ be the forest after removing the edges of $\pi=\set{p_1,\ldots,p_r}$ from $T$, and note that $|\pi|=r$.
  Furthermore, let $T_1,\ldots,T_k$ be the connected components of $T$, and let $T_1',\ldots,T_{k'}'$ be the connected components of $T'$.
  We first construct a maximal spanning forest of $G\setminus (f\cup \pi)$ (where we replace arcs in $\pi$ with their corresponding edges in $G$), and then add the edges of $f\setminus f'$ back to the forest to obtain a maximal spanning forest of $G\setminus f'$.

  To construct a maximal spanning forest of $G\setminus (f\cup \pi)$, iterate over all new components $\set{T_i'}\setminus \set{T_i}$.
  For each component, perform $O(\log n)$ cut queries in $G\setminus (f\cup \pi)$ to either find an edge between $T_i'$ and $V\setminus T_i'$ or certify one does not exist using \Cref{claim:find-outgoing-edge} and \Cref{observation:query-residual-graph}.
  Then, merge the two components and iterate until no more edges can be found.
  Notice that this process can only merge new components, as edges connecting old components do not exist in $G\setminus (f\cup \pi)\subseteq G\setminus f$.
  Finally, proceed to the next component.
  Notice that the resulting forest is maximal in $G\setminus (f\cup \pi)$ since there exist no edges in $G\setminus (f\cup \pi)$ that connect any two components of the forest.
  To bound the query complexity of the algorithm, note that each step takes $O(\log n)$ cut queries in $G\setminus (f\cup \pi)$, each of which can be simulated by a single cut query in $G$ by \Cref{observation:query-residual-graph}.
  Furthermore, each step either reduces the number of new components by $1$ or finishes processing the component.
  Since there are at most $2(r-1)$ new components formed in $T'$, the total number of queries is $\tO(r)$.

  Finally, iterate over all edges of $f\setminus f'$ and add each edge that connects two different components to the forest.
  Since the edges of $f,f'$ are explicitly known, this can be done without any cut queries.
  Notice that this forms a maximal spanning forest of $G\setminus f'$, since there exist no edges in $G\setminus f'$ that connect any two components of the forest.
\end{proof}
\begin{proof}[Proof of \Cref{thm:bounded-dinitz}]
  Begin by initializing an empty flow $f_0$, and constructing a maximal spanning forest $T$ of $G$.
  In addition, query the degrees of all vertices in $G$.
  While the flow value is at most $\delta$, if there exists an augmenting path $\pi$ in $T\cup E_{f_i}$, which we may assume is simple by removing cycles, find a short augmenting path $\pi'$ of length at most $8\cdot(n+q(f_i))/\delta$ in $G_{f_i}$ using \Cref{lem:shortcut-augmenting-path}.
  Notice that finding an augmenting path in $T\cup E_{f_i}$ does not use any cut queries, as both sets are explicitly known.
  Then, augment the flow $f_i$ along $\pi'$, to obtain $G_{f_{i+1}}$, and update $T$ to be the maximal spanning forest  of $G\setminus f_{i+1}$ using \Cref{lem:maintain-maximal-spanning-forest}.
  If no augmenting path exists, then $f_{i}$ is a maximum flow by \Cref{thm:kl98-exist-augmenting}, and the algorithm returns it.
  Finally, if the flow value is $\delta$ and there exists an augmenting path, return that the flow value is larger than $\delta$.
  
  To output a cut, apply a growing set procedure.
  Let $S=\{s\}$, and while there exists an outgoing edge from $S$ in $G_{f_i}$, add its new endpoint to $S$.
  When no such edge exists, return $S$.
  The correctness of this procedure follows from the min-cut max-flow theorem, as no augmenting $s,t$ paths remain in $G_{f_i}$.
  Note that outgoing edges from $S$ in $G\setminus f_i$ can be found using \Cref{claim:find-outgoing-edge} and \Cref{observation:query-residual-graph}, and outgoing edges that are in $f_i$ can be found trivially as it is explicitly known.
  Hence, the total number of cut queries used to find the cut is $\tO(n)$.

  To analyze the query complexity, we first show that $q(f_i)\le O(n)$ throughout the execution.
  Notice that by \Cref{lem:shortcut-augmenting-path}, we have $q(f_i) \le q(f_{i-1}) +8\cdot(n+q(f_{i-1}))/\delta$.
  Moving sides we obtain 
  \begin{equation*}
    q(f_i)+n 
    \le (1+8/\delta)(q(f_{i-1})+n)
    ,
  \end{equation*}
  which yields $q(f_i)+n \le (1+8/\delta)^i (q(f_0)+n) = O(n)$ as $i\le \delta$ and $q(f_0)=0$.
  Therefore, the total query complexity of updating the spanning forest is $\tO(n)$, and the total query complexity of finding the short augmenting paths is $\tO(n/\delta\cdot \delta)=\tO(n)$.
  In addition, the query complexity of querying all the vertex degrees is $O(n)$ and constructing the initial maximal spanning forest is $\tO(n)$ by calling \Cref{lem:ni-certificate} with $\tau=1$.
  Hence, the total query complexity of the algorithm is $\tO(n)$.
\end{proof}

\section*{Acknowledgements}
The author had the insight that the number of edges in a maximum flow should be small when the minimum degree is large, which he later realized was already known in the literature \cite{GY95,HKR97}.
The author then further observed that this should imply a version of Dinitz's algorithm that uses few phases.
This intuition was made precise and proved by ChatGPT 5.5 Pro, which gave an earlier version of \Cref{thm:bounded-dinitz}.
The author then asked ChatGPT 5.6 Ultra to improve the result, which yielded the final version of the theorem alongside its entire proof.
All proofs were manually verified and rewritten by the author, who takes full responsibility for the content of the paper.
A first draft of this paper was written by Claude Opus 4.8 and later heavily edited by the author using in-line writing suggestions from GitHub Copilot.
Finally, the author used Claude Opus 4.8 for feedback on several sections, particularly on the clarity and correctness of the English prose and the mathematical content. 

{\small
  \bibliographystyle{alphaurl}
  \bibliography{bibliography}
} %

\appendix

\section{Nagamochi-Ibaraki Sparsifier}
\label{sec:ni-sparsifier-proof}
In this section we prove \Cref{lem:ni-certificate}, this is a standard result, but we include the proof for completeness.
\begin{proof}
    The NI sparsifier is composed of a set of $\tau$ edge-disjoint spanning forests $F_1,\ldots,F_\tau$, each of the graph $G\setminus\cup_{j< i}F_j$.
    Our construction algorithm is based on iteratively peeling off spanning forests from the graph, and is based on the original construction of Nagamochi and Ibaraki~\cite{NI92}.
    We first describe how to find a single spanning forest $F_i$ of the graph $G\setminus\cup_{j< i}F_j$ using cut queries, which we then repeat $\tau$ times to find the forests $F_1,\ldots,F_\tau$.

    To find a spanning forest $F_i$, iterate over all the vertices $v\in V$.
    For each $v$, denote by $U$ the set of vertices not in the same connected component as $v$ in the current forest $F_i$.
    Then, let $U'\subseteq U$ be the set of vertices which are not neighbors of $v$ in some previously constructed forest $F_1,\ldots,F_{i-1}$.
    Find an edge from $v$ to $U'$ using $O(\log n)$ cut queries, by binary search over the set $U'$ using \Cref{claim:find-outgoing-edge}.
    Add this edge to the forest $F_i$, and repeat until no more edges are found.
    Observe that this uses $O(\log n)$ cut queries per search by \Cref{claim:find-outgoing-edge}.

    Therefore, the total query complexity is $\tO(n\tau)$, since we find at most $n-1$ edges in each of the $\tau$ forests, and for each forest we iterate once over all vertices.
    Finally, we show that this sparsifier satisfies the cut property.
    Fix some cut $S\subseteq V$, if $|E(S,V\setminus S) \cap \bigcup_{i=1}^\tau F_i|\ge \tau$, then we are done.
    Otherwise, there exists some forest $F_i$ in which $S,V\setminus S$ are disconnected.
    By the maximality of $F_i$, there exists no edge in $E(S,V\setminus S)$ that is not included in $\bigcup_{j=1}^{i-1} F_j$, and hence all edges of $E(S,V\setminus S)$ are included in the sparsifier.
    This concludes the proof.
\end{proof}

\section{Existence of Augmenting Paths in a Maximal Spanning Forest}
\label{sec:kl98-proof}
In this section we prove \Cref{thm:kl98-exist-augmenting}.
\begin{proof}[Proof of \Cref{thm:kl98-exist-augmenting}]
  Since $T\cup E_f$ is a subgraph of $G_f$, any augmenting path in $T\cup E_f$ is also an augmenting path in $G_f$.
  For the reverse direction we show that if there is no augmenting path in $T\cup E_f$, then there is no augmenting path in $G_f$.
  Assume towards contradiction that there exists an augmenting path in $G_f$, but no augmenting path in $T\cup E_f$.
  Let $S$ be the set of vertices reachable from $s$ in $T\cup E_f$.
  Notice that since $t\not\in S$, there exists an augmenting path in $G_f$ and an arc $e$ in that path that leaves $S$ in $G_f$.

  If $e\in E_f$, then it exists in $T\cup E_f$, contradicting the assumption that there is no augmenting path in $T\cup E_f$.
  Otherwise, $e$ is an arc corresponding to an edge of $G\setminus f$, and since $T$ is a maximal spanning forest of $G\setminus f$, there exists a path of undirected edges in $T$ connecting the endpoints of $e$.
  However, such a path must exit $S$ as $e$ leaves $S$, and hence there is an arc leaving $S$ in $T\cup E_f$, which leads to a contradiction.
  Therefore, there is no augmenting path in $G_f$.
\end{proof}

\end{document}